\documentclass[conference]{IEEEtran}

\ifCLASSINFOpdf

\else

\fi

\usepackage{amsmath}
\usepackage{amsthm}
\newtheorem{theorem}{Theorem}
\usepackage{amssymb}
\usepackage{algorithm}
\usepackage{algpseudocode}
\usepackage{array}
\usepackage{graphicx}
\usepackage{subfig}   
\usepackage{booktabs} 

\usepackage{cite}

\begin{document}


\title{TALSC: Timeliness-Aware Large-Small VLM Collaboration for Infrastructure-Assisted Autonomous Driving}

\author{\IEEEauthorblockN{Mengmeng Zhu, Yuxuan Sun, Wei Chen, and Bo Ai}
\IEEEauthorblockA{School of Electronic and Information Engineering,
Beijing Jiaotong University,
Beijing 100044, China\\
Email: \{mengmengzhu, yxsun, weich, boai\}@bjtu.edu.cn}
}

\maketitle

\begin{abstract}
The deployment of Vision-Language Models (VLMs) in autonomous driving (AD) systems is constrained by on-board computing power, restricting vehicles to small VLMs (SVLMs) with limited perception and reasoning capabilities. Infrastructure-assisted AD alleviates this resource constraint by enabling collaboration with large VLMs (LVLMs) at edge servers. However, in dynamic vehicular environments, the utility of sensory data for downstream tasks decays rapidly, making \emph{timeliness of information} a critical concern. To balance the accuracy gains of LVLMs with their latency-induced timeliness degradation, we develop a Timeliness-Aware Large-Small VLM Collaboration (TALSC) framework.
Specifically, we first model the Age of Information (AoI) evolution for VLM inference and characterize the coupling among AoI, token length, and task performance to formulate a general timeliness metric. 
Building on this, we propose the TALSC online scheduling algorithm. 
Since scheduling decisions have a delayed impact on future timeliness metric and the output token number is unknown at scheduling time, we design a Lyapunov drift-plus-estimated-penalty algorithm and provides a guaranteed performance.
In simulation, we first conduct a case study to derive a fitted timeliness metric based on nuScenes dataset, and further show that TALSC outperforms baselines under various communication and computing settings, achieving up to a 12.6\% normalized improvement in Micro-F1 score compared with the best-performing baseline.
\end{abstract}


%
\IEEEpeerreviewmaketitle

\section{Introduction}
The deployment of Vision-Language Models (VLMs) has shown significant potential in enhancing high-level scene understanding for autonomous driving (AD) ~\cite{zhou2026OpenDriveVLA, long2026vlmmpc}, 
supporting tasks such as advanced driving assistance~\cite{hu2025adas} and visual question answering (VQA)~\cite{fang2025guided}.
However, deploying these models on resource-constrained vehicles remains challenging. 
Due to limited onboard computation capability, vehicles can typically only accommodate small VLMs (SVLMs)~\cite{gopalkrishnan2024emvlm4ad}, whose perception accuracy is inherently limited in complex environments~\cite{zhao2025stitch}. In contrast, large VLMs (LVLMs) with massive parameters exhibit superior scene understanding and reasoning capabilities~\cite{liu2025fewshot}.
However, their high computational demands require infrastructure assistance via resource-rich edge servers. 
Such infrastructure-assisted architectures inevitably incur substantial \emph{token transmission and inference latency}~\cite{sharshar2025vlm_survey}. 

Due to the dynamic nature of vehicular environments, the utility of sensory data for downstream tasks (e.g., perception accuracy) decays rapidly over time~\cite{zhou2024task}, which is characterized by \emph{timeliness of information}~\cite{qin2023timeliness, sun2023optimizing}. 
An infrastructure-assisted AD system thus need to balance the relatively lower latency of local SVLMs and the higher accuracy of edge-assisted LVLMs to maximize the timeliness, raising two fundamental questions.


First, \emph{how to characterize timeliness in a large-small VLM collaboration framework?}
Timeliness is inherently related to the elapsed time since sensory data generation, commonly quantified by the Age of Information (AoI)~\cite{kaul2012real, qin2023timeliness, kalor2022timely}. 
However, linear AoI metrics may not capture the non-linear degradation of task performance caused by delayed information. 
Several non-linear metrics have been proposed, including Age of Usage Information~\cite{xie2023minimizing}, and Task-oriented
Age of Information~\cite{gan2025taoi}, incorporate factors such as context awareness, data usability and task relevance.  
Nevertheless, existing works lack a unified framework to characterize how information granularity (e.g., visual token length~\cite{qi2025hyperbolic} or image resolution~\cite{wang2025thinking}) interacts with communication-computation latency and ultimately affects task performance.
Finer-grained information generally improves inference accuracy~\cite{jiang2025gata}, but also incurs larger data size, resulting in increased transmission and computation latency~\cite{wu2025tokenselect, yang2026sliminfer} and thus higher AoI.
This intricate coupling makes it challenging to characterize timeliness in large-small VLM collaboration.

Second, \emph{when to invoke LVLM assistance and how to dynamically schedule resources?}
Transmitting sensory data or token embeddings to edge LVLM introduces significant communication overhead and burdens limited wireless bandwidth~\cite{you2026v2x}.
Recent advances have explored various strategies to optimize VLM collaboration. 
At the \emph{invocation decision level}, approaches such as complexity-aware invocation for hard versus easy samples~\cite{hu2024laecips} and deep reinforcement learning-based joint optimization of latency, energy consumption, and Quality of Service (QoS)~\cite{hu2025adas} have been proposed.
However, these methods typically assume a fixed communication payload per task, lacking adaptability to dynamic network conditions. 
At the \emph{communication optimization level}, prior works have investigated data compression techniques tailored for autonomous driving, such as dynamic Region-of-Interest (RoI) extraction~\cite{zhang2025vavlm} and adaptive spatial resolution scaling~\cite{zhou2025dynrsl}. 
Recent VLM architectures encode multimodal inputs into unified visual tokens~\cite{wei2026token, li2023blip2, li2023llavamed}, enabling fine-grained control over transmitted information. In addition,~\cite{qian2025edgevlm} leverages delayed but more accurate historical LVLM outputs to guide SVLM inference.
Nevertheless, these approaches primarily focus on spatial redundancy reduction, while overlooking temporal dynamics and token-level adaptation, both of which are essential for improving timeliness.

This paper proposes a Timeliness-Aware Large-Small VLM Collaboration (TALSC) framework for infrastructure-assisted AD, to jointly optimize collaboration strategy, token length and bandwidth allocation under bandwith and task deadline constraints. The main contributions are summarized as follows:
\begin{itemize}
    \item We model VLM inference latency under edge environments and characterize the coupling between AoI and token length, based on which we formulate a general timeliness metric for large-small VLM collaboration.
    \item To handle the unknown timeliness gain caused by future task latency and uncertain output token length, we develop a performance-guaranteed Lyapunov drift-plus-estimated-penalty algorithm to make VLM collaboration decisions online.
    \item We conduct extensive simulations to demonstrate that the proposed TALSC algorithm effectively balances AoI and accuracy, maximizing overall timeliness metric while satisfying long-term bandwidth and latency constraints, and consistently outperforming existing baselines.
\end{itemize}

\section{System Model}
As shown in Fig.~\ref{fig:system}, we consider an infrastructure-assisted AD system, where a vehicle equipped with an SVLM, and an edge server hosting an LVLM. 
We define a task (e.g., VQA for scene understanding) as being initiated when the vehicle captures raw sensory data, and completed when the VLM inference result is returned for downstream decision-making.
For each task, the system can either process locally using the SVLM or transmit it to the LVLM. In the LVLM mode, the raw sensory data is first encoded into visual tokens and transmitted to the edge along with the query. The LVLM processes the inputs and returns generated answers to the vehicle.

\begin{figure}
    \centering
    \includegraphics[width=0.85\linewidth]{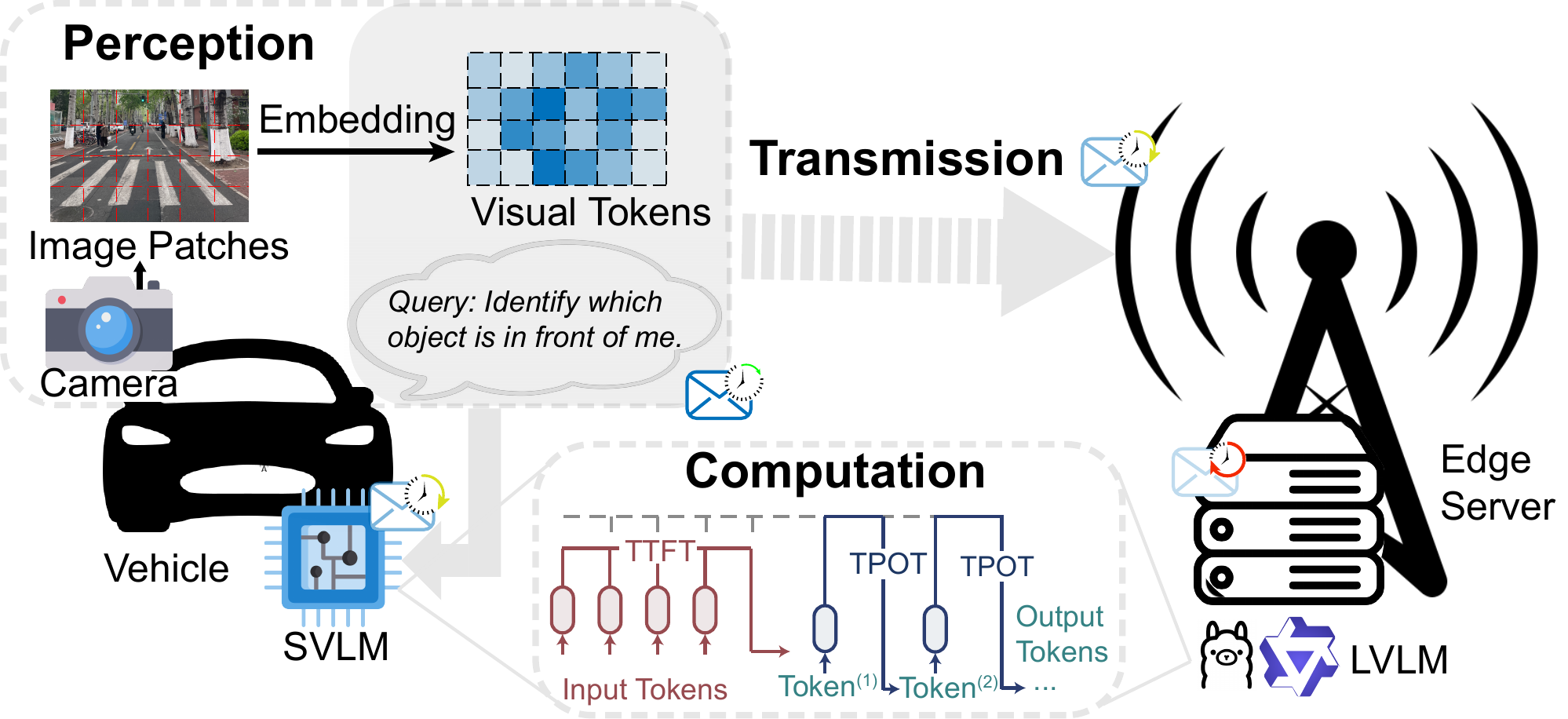}
    \caption{An infrastructure-assisted AD system where vehicle runs an SVLM and edge server hosts an LVLM. A VQA task can be either processed locally by the SVLM or uploaded for LVLM inference.
    }
    \vspace{-5mm}
    \label{fig:system}
\end{figure}

\subsection{Perception, Transmission and Computation Model }
Let $m$ denote the task index. For each task $m$, we define collaboration decision $x_m \in \{0,1\}$, where $x_m=1$ indicates LVLM inference, and $x_m=0$ indicates SVLM inference.

The raw sensory input is partitioned into patches and embedded into visual tokens. To balance inference accuracy and latency, we adopt an elastic token generation mechanism, where the visual tokens length $L_m$ can be dynamically adjusted via image resolution scaling or token pruning. The corresponding perception latency is denoted as $T_{\mathrm{per},m}$.

When $x_m=1$, the visual tokens are uploaded to the edge server. The transmission latency is given by:
\begin{equation}
T_{\mathrm{comm},m} = x_m \frac{D(L_m)}{R(b_m)},
\end{equation}
where the uplink transmission rate is $R(b_m)=b_m\log_2\left(1+\frac{P_\mathrm{tx}g_m^2}{\sigma^2}\right)$, with allocated bandwidth $b_m$, transmit power $P_\mathrm{tx}$, channel gain $g_m^2$, and noise power $\sigma^2$.
The transmitted data size is $D(L_m) = L_m d q$,
where $d$ is the token embedding dimension, and $q$ is the quantization bit-width.

For $m$-th task, the VLM inference process consists of a prefilling phase with latency $\mathrm{TTFT}_{m, x_m}$, and an autoregressive decoding phase with $k$-th output token latency $\mathrm{TPOT}_{m, x_m}^{(k)}$. Both $\mathrm{TTFT}_{m, x_m}$ and $\mathrm{TPOT}_{m, x_m}^{(k)}$ depend on the execution location $x_m$ and input resolution (e.g. token length $L_m$).
The computation latency for generating $K_m$ tokens is
\begin{equation}
T_{\mathrm{comp},m} = \mathrm{TTFT}_{m,x_m} + \sum_{k=1}^{K_m} \mathrm{TPOT}_{m,x_m}^{(k)},
\end{equation}
The downlink transmission latency for the generated response is assumed to be negligible.
The total latency of task $m$ is 
\begin{equation}
T_{\mathrm{total},m} \triangleq T_{\mathrm{per},m} + x_m T_{\mathrm{comm},m} + T_{\mathrm{comp},m}.
\end{equation}

Tasks arrive periodically at intervals of $\Delta T$, representing, for example, VQA queries triggered by high-level driving decisions. For transmission stability, we impose a long-term average bandwidth constraint $\bar{B}$. To guarantee driving safety, each task is subject to a hard latency constraint $T_\mathrm{req}$.

\subsection{Timeliness Metric for SVLM and LVLM Collaboration}
Let $t$ denote the continuous time variable, and let $t_m$ and $t_m'$ denote the generation time and completion time of task $m$, respectively. Let $h_t$ denote the AoI at time $t$. The AoI increases linearly with time and is reset upon the completion of each task~\cite{kaul2012real}, as illustrated in Fig.~\ref{fig:aoi_utility}:
\begin{equation}
    h_t = t - t_m,~~ m=\max\{m\mid t_m' <t\}.
\end{equation}

\begin{figure*}[!t] 
    
    \begin{minipage}{0.48\textwidth}
        \centering
        \includegraphics[width=0.95\linewidth]{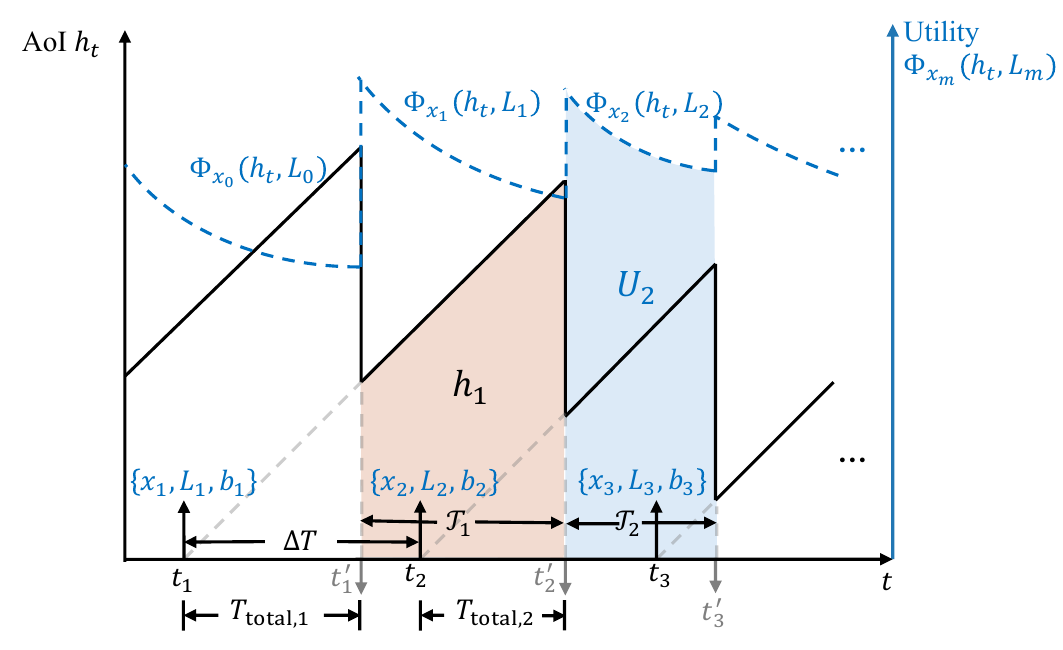}
        \caption{Evolution of AoI and the proposed timeliness metric. 
        The AoI $h_t$ increases linearly over time and is reset upon task completion. 
        The $m$-th task starts at $t_m$ with decisions $\{x_m, L_m, b_m\}$ and completes at $t'_m$. 
        The interval $\mathcal{T}_m$ is defined as the duration between two consecutive task completions, i.e., from $t'_m$ to $t'_{m+1}$.}
        \label{fig:aoi_utility}
    \end{minipage}
    \hfill 
    \begin{minipage}{0.48\textwidth}
        \begin{algorithm}[H]
            \caption{The TALSC Algorithm}
            \label{alg}
            \begin{algorithmic}[1]
            \State \textbf{Initialize:} Virtual queue $Q_1=0$, trade-off parameter $V$, task interval $\Delta T$, bandwidth budget $\bar{B}$.
            \For{each task $m \in \mathcal{M}$}
            \State \textbf{Case 1:} Assume $x_m = 1$ (LVLM inference)
                \For{$L_m = \ell_k \in \mathcal{L}$}
                    \State Calculate $b_m^{(k)}$ according to \eqref{solveb}.
                    \If{$T_{\mathrm{total},m}(\ell_k)\leq T_\mathrm{req}$} 
                        \State Calculate \eqref{edpp_obj} as $\mathcal{J}_m^{(k)}$.
                    \EndIf
                \EndFor
            \State \textbf{Case 2:} Assume $x_m = 0$ (SVLM inference)
                \State Calculate \eqref{edpp_obj} as $\mathcal{J}_m^{(0)}$.
            
            \State Get $k^* = \arg\min_k \mathcal{J}_m^{(k)}$.
            \State Let $\{x_m^*, b_m^*, L_m^*\} \gets \{x_m^{(k)}, b_m^{(k)},L_m^{(k)}\}$.
            \State Update $\tilde{h}_{m+1}$ and $Q_{m+1}$ according to \eqref{est_aoi}, \eqref{vq}.
            \EndFor
            \end{algorithmic}
        \end{algorithm}
    \end{minipage}
\end{figure*}
The task performance relies fundamentally on the inference model type, determined by $x_{m}$. Furthermore, the performance is jointly influenced by the information granularity, dictated by the number of visual tokens $L_{m}$, and the information freshness, characterized by $h_t$. Accordingly, the timeliness metric at time $t$ can be expressed as $\Phi_{x_m}(h_t, L_m)$. 

\subsection{Problem Formulation}
Let $\mathcal{M}=\{1, 2, \cdots, M\}$ denote the set of sequential perception tasks generated by the vehicle. 
For each task $m \in \mathcal{M}$, we aim to jointly optimize collaboration decision $x_m \in \{0, 1\}$, visual token length $L_m \in \mathcal{L}= \{\ell_1, \cdots, \ell_\mathrm{max}\}$, and uplink bandwidth allocation $b_m \in (0, B_{max}]$. 
We maximize the long-term average timeliness metric, bounded by long-term average bandwidth budget $\bar{B}$ and hard latency constraint $T_\mathrm{req}$.
\begin{subequations}
\begin{align}
    \mathcal{P}1: \max_{\{x_m, L_m, b_m\}} \quad & \lim_{T \to \infty} \frac{1}{T} \int_0^T \Phi_{x_{m}}(h_t, L_{m})\,dt  \\
    \text{s.t.} \quad & \lim_{M \to \infty} \frac{1}{M} \sum_{m=1}^{M} b_m \cdot x_m \leq \bar{B}, \\
    & T_{\mathrm{total}, m} \leq T_\mathrm{req}, \quad \forall m \in \mathcal{M}. \label{Treq}
\end{align}
\end{subequations}

Solving Problem $\mathcal{P}1$ faces three challenges. First, the schedule decisions have \emph{delayed and cumulative} impact on the future AoI and timeliness metric evolution, we need to reveal this impact mechanism to guide scheduling.
Second, the latency of the next task and the output tokens number of current task are required to determine the next AoI update, but remain unknown at the decision time thus need to be \emph{estimated}.
Third, the problem is \emph{stochastic}, involving both \emph{long-term average} objective and constraint, requiring online decision-making without prior knowledge of future states.


\section{Timeliness-Aware Large-Small VLM Collaboration Scheduling Algorithm}
To solve Problem $\mathcal{P}1$, we first decompose the long-term average objective into intervals determined by each scheduling decision. Then we characterize the cumulative timeliness metric of each interval by estimating the output token number and latency of the next task, and further propose a Lyapunov drift-plus-estimated-penalty online scheduling algorithm. Finally, we analyze the performance of the proposed algorithm which is guaranteed by the offline optimal problem. 

\subsection{Online Scheduling Algorithm}


Recall that $t_m$ denotes the start time and $t_m'$ denotes the complete time of task $m$. Define $\mathcal{T}_m \triangleq [t_m',t_{m+1}')$ as the interval between two consecutive inferences.
Let $U_m \triangleq \int_{\mathcal{T}_m} \Phi_{x_{m}}(h_t, L_{m})\,dt$. Then $\mathcal{P}1$ can be reformulated as
\begin{equation}
   \mathcal{P}2:
   \max_{\{x_m, L_m, b_m\}} \lim_{M \to \infty} \frac{1}{M \Delta T} \sum_{m=1}^{M} U_m 
\end{equation}


Let $h_m \triangleq \frac{1}{\left|\mathcal{T}_m \right|}\int_{\mathcal{T}_m}h_t\,dt$.
The timeliness metric can be expanded around $h_m$ using Taylor expansion:
\begin{align}
    \Phi_{x_{m}}(h_t, L_{m}) = 
    & \Phi_{x_{m}}\left(h_m, L_{m}\right)  + \Phi_{x_{m}}'\left(h_m, L_m\right)\cdot\left(h_t - h_m\right) \nonumber \\
    & + \frac{1}{2}\Phi_{x_{m}}''\left(\xi, L_m\right)\cdot\left(h_t - h_m\right)^2, \label{taylor}
\end{align}
where $\xi$ is a value between $h_m$ and $h_t$. Let $\epsilon_t$ denote the last term of \eqref{taylor}. 
Let $\epsilon_m \triangleq \frac{1}{\left|\mathcal{T}_m \right|}\int_{\mathcal{T}_m}\epsilon_t\,dt$.
Hence, $U_m =  \left|\mathcal{T}_m\right|\cdot \left(\Phi_{x_m}(h_m, L_m)+\epsilon_m\right)$.

At each scheduling time $t_m$, $U_m$ remains \emph{unknown} since the total latency $T_{\mathrm{total},m}$ is uncertain due to the unknown number of output tokens $K_m$. The number of output tokens $K_m$ is estimated by its statistical expectation, denoted as $\tilde{K}_m$. 
Based on the geometric structure in Fig.~\ref{fig:aoi_utility} we have:
\begin{equation}
    \tilde{h}_m  \triangleq \frac{1}{2}\left(T_{\mathrm{total},m}(\tilde{K}_m)+\left|\mathcal{T}_m\right|\right)^2 \!\!\!- \frac{1}{2}\left(T_{\mathrm{total},m}(\tilde{K}_m)\right)^2 \!\!\!.
    \label{est_aoi}
\end{equation}
Denote the estimation error of the timeliness metric as $\delta_m \triangleq \Phi_{x_m}(\tilde{h}_m, L_m) - \Phi_{x_m}(h_m, L_m)$. 
Moreover, for $|\mathcal{T}_m|\triangleq \Delta T - T_{\mathrm{total},m} + T_{\mathrm{total},m+1}$, the total latency of next task $T_{\mathrm{total},m+1}$ remains unknown. For analytical tractability, we approximate $|\tilde{\mathcal{T}}_m|$ by $\Delta T$, and treat the deviation as a bounded estimation error $\zeta_m \triangleq T_{\mathrm{total},m} - T_{\mathrm{total},m+1}$.
Denote that $\tilde{U}_m \triangleq  |\tilde{\mathcal{T}}_m| \cdot \Phi_{x_m}(\tilde{h}_m, L_m)$. 
Thus, we have:
\begin{equation}
    \tilde{U}_m  = \left(\left|\mathcal{T}_m\right| + \zeta_m\right)\cdot \left(\Phi_{x_m}(h_m, L_m)+\delta_m\right)
     \leq U_m + C,
\end{equation}
where the constant $C \triangleq 2(\delta_\mathrm{max}+2\epsilon_\mathrm{max})\Delta T+\zeta_\mathrm{max} \cdot (\Phi_\mathrm{max}+\delta_\mathrm{max})$. Let $\delta_\mathrm{max} \triangleq \max\{\left|\delta_m\right|\}$, $\epsilon_\mathrm{max} \triangleq \max\{\left|\epsilon_m\right|\}$, $\zeta_\mathrm{max} \triangleq \max\{\left|\zeta_m\right|\}$. Let $\Phi_\mathrm{max}$ be the maximum available timeliness.

To enforce the bandwidth constraint, define the virtual queue
\begin{equation}
    Q_{m+1}=\max\{Q_m+x_m b_m-\bar{B},0\}, \label{vq}
\end{equation}
with $Q_1=0$.
Inspired by the drift-plus-penalty algorithm of Lyapunov optimization \cite{neely2010stochastic}, the online scheduling aims to solve the \emph{estimated-drift-plus-penalty} problem:
\begin{subequations}
\begin{align}
    \mathcal{P}3: \min_{\{x_m, L_m, b_m\}} & x_m Q_m b_m - V \cdot \tilde{U}_m \label{edpp_obj}\\
    \text{s.t.} ~~~~ & T_{\mathrm{total}, m} \leq T_\mathrm{req}, \quad \forall m \in \mathcal{M}. 
\end{align}
\end{subequations}

As $\mathcal{P}3$ is a mixed-integer nonlinear programming problem, we exploit the discrete nature of $\{x_m, L_m\}$ and decompose it into two stages without loss of optimality. 
For $k$-th feasible pair $\{x_m, L_m\}$, we first solve the subproblem by
\begin{equation}
    b_m^{(k)}=\arg\min_{b} Q_m b - V \Delta T \Phi_1\left(\tilde{h}_m, \ell_k\right). \label{solveb}
\end{equation}
Substituting $b_m^{(k)}$ into $\mathcal{P}3$, we enumerate all possible combinations of $\{x_m, L_m\}$ to find the optimal solution $\{x_m^\star, L_m^\star, b_m^\star\}$ that maximizes the objective $\mathcal{J}_m^{(k)}$.
The proposed TALSC algorith is summarized in Algorithm~\ref{alg}.

\subsection{Performance Analysis}
The performance of the proposed scheduling algorithm is characterized by comparing with its optimal offline counterpart $\mathcal{P}2$, and is shown in the following theorem.

\begin{theorem}
Compared to the offline optimal solution, the proposed algorithm achieves a performance guarantee:
\begin{equation}
\sum_{m=1}^M U_m^\ddagger 
\ge \sum_{m=1}^M U_m^* - \frac{z_0 M}{V} + 2MC,
\end{equation}
where $U_m^*$ denotes the optimal offline timeliness, and $z_0, C$ are bounded constants.

Moreover, the bandwidth virtual queue can be bounded by:
\begin{equation}
\begin{aligned} 
    &\sum_{m=1}^M \left(x_m b_m -\bar{B}\right) 
    &\leq \sqrt{2V\sum_{m=1}^M (U_\mathrm{max}+C)+z_0M}.
\end{aligned}
\end{equation}
\end{theorem}

\begin{proof}
Define the Lyapunov function as $\mathcal{L}(Q_m)=\frac{1}{2}Q_m^2$, and the $M$-task drift as $\Delta_M = \mathbb{E}\left[\mathcal{L}(Q_{M+1})-\mathcal{L}(Q_1) \right]$.
Let $^*$ and $^\ddagger$ denote the optimal offline and proposed policies, respectively.
The $M$-task drift-plus-penalty function is bounded by:
\begin{equation}
\begin{aligned}
    &\Delta_M^\ddagger - V \cdot \sum_{m=1}^M U_m^\ddagger \\
    & \leq z_0 M + \sum_{m=1}^M \mathbb{E}\left[Q_m y_m^*\right]
    - V \cdot \sum_{m=1}^M \tilde{U}_m^*
    + VMC  \\
    & \leq z_0 M 
    - V \cdot \sum_{m=1}^M U_m^*
    + 2VMC.
\end{aligned}
\end{equation}
where the first inequality follows from the optimality of $\mathcal{P}3$, and the second uses 
$\mathbb{E}[y_m^*] \le 0$, which implies $\mathbb{E}[Q_m y_m^*] \le 0$.
Since $\Delta_M \ge 0$, rearranging terms completes the proof.

Next, summing the queue dynamics gives
\begin{equation}
\sum_{m=1}^M (x_m b_m - \bar{B})
\le Q_{M+1}.
\end{equation}
Using $\mathcal{L}(Q_{M+1}) = \frac{1}{2}Q_{M+1}^2 \le \Delta_M$, we obtain
\begin{equation}
Q_{M+1} \le \sqrt{2\Delta_M}
\le \sqrt{2V M (U_{\max}+C) + z_0 M}.
\end{equation}

\end{proof}

\section{Simulation Results}
\subsection{Empirical Experiments} \label{emp_exp}
We conduct a case study on the timeliness metric using nuScenes dataset~\cite{Caesar_2020_CVPR}. We extract 400 image samples from 10 distinct driving clips. Each clip lasts about 20 seconds sampled at 2 Hz. To emulate the AoI, we introduce a temporal offset between the input images and ground-truth annotations.
We deploy Qwen2.5-VL-3B and Qwen2.5-VL-72B as the onboard SVLM and the edge-based LVLM, respectively. Internally, the Qwen-VL architecture encodes every $28 \times 28$ pixel patch into a single visual token. Therefore, we obtain elastic visual token lengths $\mathcal{L} = \{32, 50, 120, 231, 480\}$ by scaling the input image resolution.
We evaluate performance via a multi-object classification task. The task is to identify the 10 core object categories in nuScenes dataset in each scene. We use the Micro-F1 score to reflect overall classification accuracy.

\begin{figure*}[!t]
\centering

\begin{minipage}{0.32\textwidth}
\centering

\subfloat[Qwen2.5-VL-3B]{
    \includegraphics[width=\linewidth]{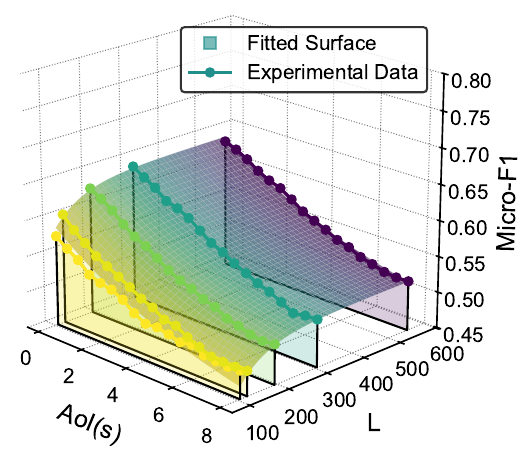}
    \label{fig:3B}
}


\subfloat[Qwen2.5-VL-72B]{
    \includegraphics[width=\linewidth]{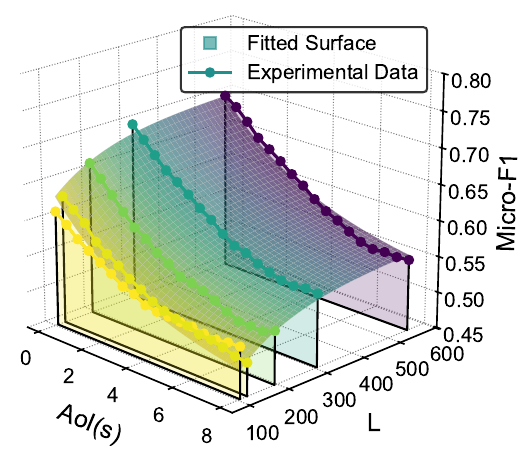}
    \label{fig:72B}
}

\captionof{figure}{Empirical fitted timeliness metric.}
\label{fig:left}
\end{minipage}
\hfill
\begin{minipage}{0.64\textwidth}
\centering

\subfloat[Token output dynamics across tasks.]{
    \includegraphics[width=0.465\linewidth]{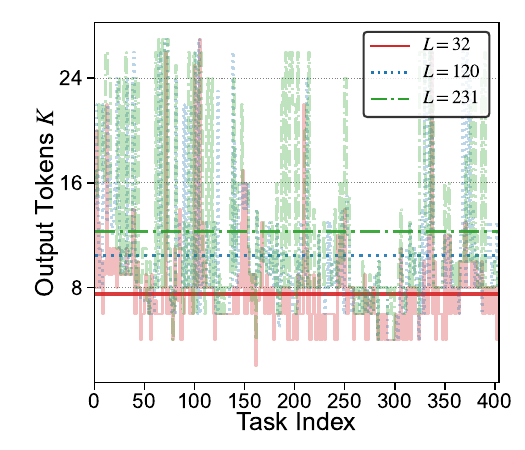}
    \label{fig:Kout}
}
\hfill
\subfloat[Impact of bandwidth budget $\bar{B}$.]{
    \includegraphics[width=0.465\linewidth]{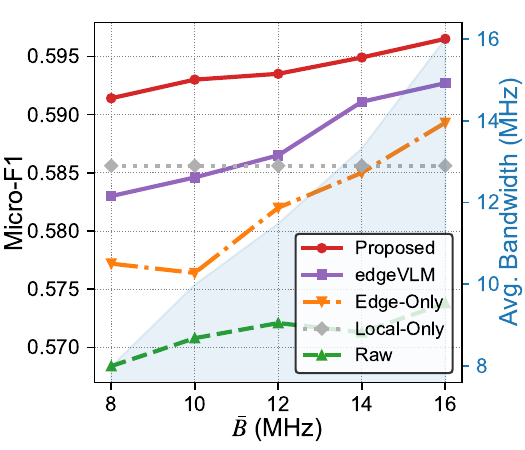}
    \label{fig:B_bar}
}


\subfloat[Impact of edge-to-local latency ratio.]{
    \includegraphics[width=0.465\linewidth]{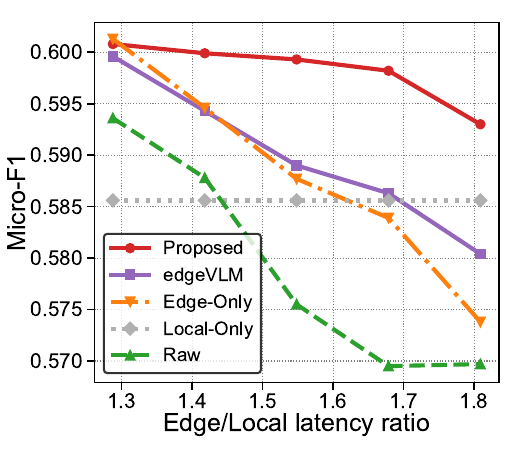}
    \label{fig:scale}
}
\hfill
\subfloat[Performance under different SNRs.]{
    \includegraphics[width=0.465\linewidth]{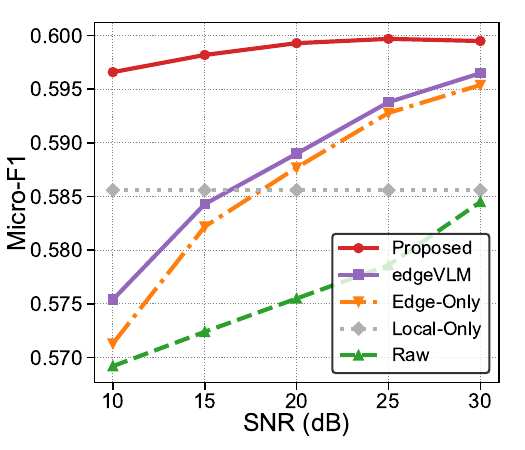}
    \label{fig:snr}
}

\captionof{figure}{Output token number dynamics and performance comparison under different bandwidth constraints, latency trade-offs, and average SNR.}
\label{fig:right}
\end{minipage}
\end{figure*}

The empirical timeliness data reveal the relationship among the task performance (measured by Micro-F1), AoI, and visual token length, along with the corresponding fitting curves, are illustrated in Fig.~\ref{fig:3B} and Fig.~\ref{fig:72B}. We observe the clear marginal diminishing returns with respect to both token length scaling and information freshness, and adopt a double exponential decay model to fit the data. Thus, for each task the timeliness metric is formulated as:
\begin{equation}
    \Phi(h, L) = \alpha^x  -\beta^x \cdot e^{-\mu^x L} + \delta^x e^{-\lambda^x h},
\end{equation}
where $L$ is the visual token length, $h$ is the AoI, and $\alpha^x, \beta^x, \mu^x, \delta^x, \lambda^x$ are the fitting parameters depending on the used model $x$. The fitted parameters are summarized in Table~\ref{tab:fitting_params}.

\begin{table}[!h]
\renewcommand{\arraystretch}{1.3}
\caption{Fitted Parameters for the Timeliness Metric}
\label{tab:fitting_params}
\centering
\begin{tabular}{c|ccccc}
\hline
\textbf{Model ($x$)} & $\alpha^x$ & $\beta^x$ & $\mu^x$ & $\delta^x$ & $\lambda^x$ \\
\hline
Qwen2.5-VL-3B  & 0.3227 & 0.1088 & 0.0183 & 0.3149 & 0.0700 \\
Qwen2.5-VL-72B & 0.4600 & 0.1106 & 0.0081 & 0.2429 & 0.1369 \\
\hline
\end{tabular}
\end{table}

\subsection{Simulation Setup}
We evaluate the perception latency $T_\mathrm{per}$ and computation latency $T_\mathrm{comp}$ using Qwen2.5-VL-3B on an NVIDIA RTX A4000 GPU as the local SVLM.
For the edge LVLM, we use the Alibaba Qwen2.5-VL-72B API. When tasks are offloaded to the edge, visual tokens are extracted locally, so that $T_\mathrm{per}$ is identical to the SVLM case. Empirical results show that $T_\mathrm{per}$ is approximately linear with respect to the visual token length $L$, which can be modeled as
\[
T_\mathrm{per} = \theta_\mathrm{per} L + \tau_\mathrm{per},
\]
where $\theta_\mathrm{per} = 0.0441$ ms/token and $\tau_\mathrm{per} = 38.10$ ms.
For VLM inference, both TTFT and TPOT scale approximately linearly with the input token length $L$ and the output token number $K$, respectively. Therefore, the computation latency is
\[
T_\mathrm{comp} = \theta_\mathrm{comp} L + \tau_\mathrm{comp} + \tau_\mathrm{out} K.
\]
The empirical parameters are summarized in Table~\ref{tab:latency}.

Furthermore, the transmission data is determined by embedding dimension $d = 8192$ and quantization bit width $q = 16$ bits. The communication link is modeled as a Rayleigh fading channel, with the default average SNR 20 dB.

\begin{table}[!t]
\renewcommand{\arraystretch}{1.2}
\caption{Empirical Parameters for Computation Latency}
\label{tab:latency}
\centering
\begin{tabular}{lccc}
\toprule
\textbf{Model} & $\theta_\mathrm{comp}$ (ms/token) & $\tau_\mathrm{comp}$ (ms) & $\tau_\mathrm{out}$ (ms) \\
\midrule
Local SVLM (3B)  & 0.285 & 121.59 & 66.10 \\
Edge LVLM (72B)  & 0.262 & 933.24 & 50.00   \\
\bottomrule
\end{tabular}
\end{table}

\subsection{Performance Evaluation}
We compare the TALSC algorithm against four baselines:
\begin{itemize}
    \item \textbf{Raw:} Dynamically schedules tasks under the bandwidth budget but processes uncompressed visual tokens, acting as an ablation for token length elasticity.
    \item \textbf{Local-Only:} Executes all tasks on onboard SVLM with elastic token lengths.
    \item \textbf{Edge-Only:} Offloads all tasks to edge LVLM with elastic token lengths, subject to the bandwidth budget.
    \item \textbf{edgeVLM~\cite{qian2025edgevlm}:} A novel edge-cloud collaboration framework. It falls back to local execution only upon consecutive dropouts, while incorporating token elasticity under the bandwidth budget.
\end{itemize}

Fig.~\ref{fig:Kout} illustrates the dynamics of the number of output tokens for each task. Since the exact output token number is unknown at scheduling time, we use its expectation for scheduling decisions. We further note that as $L$ increases, the expected number of output tokens also increases.

Fig.~\ref{fig:B_bar} illustrates the time-averaged Micro-F1 performance under different bandwidth budgets $\bar{B}$. 
Both \textit{edgeVLM} and \textit{Edge-Only} degrade under small $\bar{B}$ due to bandwidth scarcity, leading to latency violations and task drops. In contrast, \textit{Local-Only} is insensitive to $\bar{B}$ but limited by SVLM capability.
The blue shaded region indicates the average bandwidth of TALSC, showing that the proposed algorithm can satisfy the bandwidth limit $\bar{B}$.


Fig.~\ref{fig:scale} evaluates the impact of the latency ratio between edge and local processing. As the ratio decreases, the benefit of edge inference diminishes, leading to more tasks being processed locally. As a result, the performance gap between TALSC and the baselines, including \textit{edgeVLM} and \textit{Edge-Only}, becomes more pronounced.

Fig.~\ref{fig:snr} shows system performance under varying SNR levels. Under poor channel conditions, both \textit{Edge-Only} and \textit{Raw} degrade due to transmission overhead. TALSC maintains stable performance by adapting token length and favoring local inference. As SNR increases, it shifts toward edge inference and consistently outperforms \textit{EdgeVLM}, while avoiding congestion-induced inefficiency.
The SVLM and LVLM upper bounds under ideal conditions are 0.6312 and 0.6972, respectively. TALSC achieves a 12.6\% normalized improvement over the feasible performance range at SNR = 10 dB.


\section{Conclusion}
In this paper, we studied the timeliness in large-small VLM collaboration for autonomous driving. We characterized the AoI evolution in VLM inference and established a general timeliness metric.
We proposed TALSC, a Lyapunov drift-plus-estimated-penalty based algorithm, which handles delayed scheduling impact and unknown output token number. The proposed framework reveals the coupling between AoI dynamics and long-term task performance, and provides provable performance guarantees.
Simulation results demonstrate that TALSC significantly outperforms existing baselines while satisfying bandwidth and latency constraints. In future work, we plan to extend this collaborative framework to multi-user vehicular networks, further exploring semantics-aware token selection and multi-agent cooperative VLM perception.

\bibliographystyle{IEEEtran}
\bibliography{conf_refs}

\end{document}